\documentclass{article}
\pdfoutput=1
\usepackage[utf8]{inputenc}
\usepackage[hyphens]{url}
\usepackage{enumerate}
\usepackage[maxnames=8]{biblatex}
\usepackage{amsmath}
\usepackage{amssymb}
\usepackage{amsthm}
\usepackage{mathtools}
\newcommand{\GG}{\mathcal{G}}
\newcommand{\A}{\mathcal{A}}
\newcommand{\unif}{\mathcal{U}}
\newcommand{\tok}{\textit{Tok}}
\newcommand{\shuff}{\textit{Shuff}}
\newcommand{\defeq}{\vcentcolon=}
\newcommand{\iud}{IUD}
\newcommand{\compequiv}{\stackrel{c}{\equiv}}

\title{Robust ambiguity for contact tracing}
\author{David Mestel\footnote{This work was supported by FNR under project SmartExit.}\\University of Luxembourg\\\nolinkurl{david.mestel@uni.lu}}
\date{}

\begin{document}

\newtheorem{theorem}{Theorem}
\maketitle

\begin{abstract}
    A known drawback of `decentralised' contact tracing architectures is that users who have been in contact with an infected person are able to precisely identify the relevant contact, and thereby perhaps identify the infected person.  In \cite{chan2020pact}, the PACT team discuss a simple DH-based protocol to mitigate this problem, but dismiss it because it is vulnerable to a malicious user who may deviate from the specified behaviour.  This note presents a modified protocol which achieves robustness against a fully malicious user, and establishes some simple security properties.
\end{abstract}

\section{Introduction}

In the design of contact tracing apps, the choice between `centralised' and `decentralised' architectures has received great public attention.  The latter, which has been adopted by the majority of countries, by the Google-Apple API and by the DP3T consortium~\cite{troncoso2020decentralized}, has many privacy advantages, but one disadvantage is that each individual user is able to determine which of the many tokens they have collected came from an infected user, and consequently (by recalling the precise time and strength of the contact), may be able to identify the infected individual among their contacts~\cite{vaudenay2020centralized}.  This may compromise the privacy of the infected person, and violates the principle of manual contact tracing that a person should be told only that they have been in contact with an infected person, and not the person's identity~\cite{cdc2020contact}.

One possible solution to this problem, discussed in \cite{chan2020pact}, is for infected users to send to the system the tokens they have collected rather than those they have distributed.  These can then be rerandomised before being broadcast to all users, so that users are able to recognise a rerandomised token as being derived from one they have broadcast, but not specifically which one.  The simple protocol described in \cite{chan2020pact} is as follows (working in a multiplicative group $\GG$ of prime order with generator $g$):
\begin{enumerate}
    \item Each user $u$ generates a keypair $(s_u,g^{s_u})$ with random $s_u$.
    \item Each broadcast token is of the form $t_i=(g^{\alpha_i},g^{\alpha_is_u})$, where each $\alpha_i$ is chosen uniformly at random.
    \item On testing positive, a user who has received the token $t=(x,y)$ uploads $(x^\beta,y^\beta)$ for fresh random $\beta$.
    \item To determine whether they are at risk, user $u$ checks whether a token of the form $(x,y)$ with $y=x^{s_u}$ is present on the server.
\end{enumerate}

Conditional on the Decisional Diffie-Hellman (DDH) assumption on $\GG$, the tokens generated by each user are pseudorandom, and furthermore a user who learns they are at risk cannot tell which of their tokens (generated using $s_u$) was reported, because of the exponentiation by fresh random $r$.  However, as the authors note, this protocol is fatally flawed in the presence of malicious users.  Such a user can simply use a fresh $s_u$ for every token they generate, and this cannot be detected or prevented (one could imagine requiring every token to be accompanied by a zero-knowledge proof that it was generated using one of the keys in a public list of public keys, but this would not be remotely practical).  

In this note we show how to modify the protocol so as to be robust against fully malicious users, at the cost of requiring the server to send a `personalised' (but non-secret) set of rerandomised tokens to each user.  The essential idea is to extend the rerandomisation step such that the messages corresponding to malformed tokens are flat random. 

\section{Protocol description}

We describe the protocol in three phases: registration; broadcasting, where a user has contacts with others and transmits tokens; and infection, after the user has tested positive.  Throughout, $\GG$ is assumed to be a multiplicative abelian group of prime order $p$, with generator $g$.

\begin{enumerate}[-]
\item\textbf{Registration phase:} sample $s\leftarrow \unif([0,p-1])$ and send $g^s$ (non-anonymously) to the server, which adds it to the list of public keys.

Since registration is not required to be anonymous, the server can ensure that each individual is only able to register a single key.

\item\textbf{Broadcast phase:} sample $\alpha \leftarrow \unif([1,p-1])$, and broadcast the token \[\tok(s,\alpha) \defeq (g^\alpha,g^{s\alpha}),\]
replacing $\alpha$ with a fresh random value after a suitable period.

\item\textbf{Infection phase:} a user who tests positive sends the server the list of tokens it has received.  The server verifies that each token $t=(x,y)$ has $x\neq g$ (discarding those that fail).  At the end of each day, for each user $u$, say with public key $g^{s_u}$, for each token $t=(x,y)$ in its list the server samples $\beta,\gamma \leftarrow \unif([0,p-1])$ and sends to $u$
\[\shuff((x,y),\beta,\gamma,g^{s_u}) \defeq (x^\beta g^\gamma,y^\beta (g^{s_u})^\gamma).\]
On receiving $t'=(z,w)$, user $u$ checks whether $w=z^{s_u}$, and if so knows that they are infected.
\end{enumerate}

Note that ambiguity (but not other privacy properties) is dependent on the honesty of the server; similarly, in manual contact tracing ambiguity is dependent on the discretion of the tracer.

To establish correctness, observe that if $t'$ came from a token broadcast by $u$, then we have
\begin{align*}
t'&=\shuff\left(\tok(s,\alpha),\beta,\gamma,g^{s_u}\right)\\
&=\left((g^\alpha)^\beta g^\gamma,(g^{s_u\alpha})^\beta (g^{s_u})^\gamma\right)\\
&=\left(g^{\alpha\beta+\gamma},g^{s_u(\alpha\beta+\gamma)}\right),
\end{align*}
as required (for some $\alpha,\beta,\gamma$).

The communication cost of this protocol is equivalent to that of just sending each user a list of all the tokens from infected users, with $O(1)$ computational cost per message passed from server to user (a single exponentiation by the user, and four exponentiations by the server).

\section{Security properties}

In this section we establish three key security properties of the protocol.  First, conditional on the DDH assumption on $\GG$, the tokens broadcast by a user with randomly chosen key are computationally indistinguishable from independent random group elements, even with knowledge of the public key (Theorem \ref{thm:indist}), and so no privacy is lost by uninfected users.  Second, for each user $u$ the output of $\shuff$ (as a probability distribution on $\GG\times\GG$ with random $\beta,\gamma$) is equal on all tokens honestly generated by $u$ (Theorem \ref{thm:honest}), and so an honest-but-curious $u$ will be unable to determine which of their broadcast tokens corresponded to contact with an infected person.  Third, for any $u$ the output of $\shuff$ on any input other than a token honestly generated by $u$ is uniformly random (Theorem \ref{thm:dishonest}), and so a malicious user is not able to defeat ambiguity by broadcasting malformed tokens.

\begin{theorem}\label{thm:indist}
Let $k$ be a positive integer, and $(S, A_1,\ldots,A_k)\sim \unif([0,p-1]\times[0,p-1]^k)$.  If $\GG$ satisfies the DDH assumption then \[(g^S,\tok(S,A_1),\ldots,\tok(S,A_k)) \compequiv \unif(\GG\times (\GG\times \GG)^k).\]
\end{theorem}
\begin{proof}
Let $\A$ be a PPT algorithm distinguishing the two distributions, and let $(g_1,g_2,g_3)$ be a DDH challenge (so either $(g_1,g_2,g_3) = (g^x,g^y,g^{xy})$ or $(g^x,g^y,g^z)$ for random $x,y,z$).  Run $\A$ on $(g_1,(g_2^{s_1},g_3^{s_1}),(g_2^{s_2},g_3^{s_2}),\ldots,(g_2^{s_k},g_3^{s_k}))$ for random $s_1,s_2,\ldots,s_k\in [0,p-1]$.
\end{proof}

\begin{theorem}\label{thm:honest}
Let $(X,Y)\sim \unif([0,p-1]^2)$ and $Z\sim \unif(\GG)$.  Then for all $s\in [0,p-1], \alpha\in [1,p-1],$ we have
\[\shuff\left(\tok(s,\alpha),X,Y,g^s\right) \equiv (Z,Z^s).\]
\end{theorem}
\begin{proof}
We have $\shuff(\tok(s,\alpha),\beta,\gamma,g^s) = (g^{\alpha\beta+\gamma},(g^{\alpha\beta + \gamma})^s)$.  Since $\gamma$ is uniformly distributed, so is $\alpha\beta+\gamma$ and hence so is $g^{\alpha\beta+\gamma}$, as required.
\end{proof}

\begin{theorem}\label{thm:dishonest}
Let $(X,Y)\sim \unif([0,p-1]^2)$.  Then for all $s\in [0,p-1], t\in (\GG\setminus \{e\})\times \GG,$ either $t = \tok(s,\alpha)$ for some $\alpha\in [1,p-1]$ or we have
\[\shuff\left(t,X,Y,g^s\right) \equiv \unif(\GG\times\GG).\]
\end{theorem}
\begin{proof}
Without loss of generality $t=(g^\alpha,g^{s'\alpha})$ for some $s'$ and $\alpha\neq 0$.  If $t\neq \tok(s,\alpha)$ then $s'\neq s$.  Then 
\begin{align*}
\shuff(t,\beta,\gamma,g^s) &= \left((g^\alpha)^\beta g^\gamma,(g^{s'\alpha})^\beta (g^s)^\gamma\right)\\
&=\left(g^{\alpha\beta + \gamma},g^{s'\alpha\beta + s\gamma}\right)\\
&=\left(g^{\alpha\beta + \gamma},g^{s'(\alpha\beta + \gamma) + (s-s')\gamma}\right).
\end{align*}
Since $\alpha \neq 0$ and $\beta,\gamma$ are independently uniformly distributed (\iud), we have that $\alpha\beta$ and $\gamma$ are \iud, and hence so are $\alpha\beta + \gamma$ and $\gamma$.  Hence since $s-s'\neq 0$ we have that $\alpha\beta + \gamma$ and $s'(\alpha\beta+\gamma)+(s-s')\gamma$ are IUD and hence so are $g^{\alpha\beta + \gamma}$ and $g^{s'(\alpha\beta + \gamma) + (s-s')\gamma}$, as required.
\end{proof}

\section{Discussion}

\subsection*{Related work}

The other approach for achieving ambiguity of which the author is aware is to use a Private Set Intersection Cardinality (PSI-CA) protocol to allow users to determine whether the set of tokens they have collected intersects with the set of tokens held by the server from infected users, without learning which tokens are in the intersection.  This was proposed independently in \cite{trieu2020epione} and in \cite{contrail2020}.  The security analysis in \cite{trieu2020epione} is expressly limited to the semi-honest setting, although it is suggested that one could guard against a dishonest user by requiring them to provide zero-knowledge proofs of correct behaviour, no doubt with significant performance consequences.

The protocol in \cite{contrail2020} is similarly clearly flawed in the presence of a fully malicious user (specifically, at step 2 of the protocol, Alice may use different values of $\alpha$ for different $x_i$ and thereby reidentify elements despite Bob's permutation).  Moreover, no proofs are provided for the claimed security properties, and it seems that even in the semi-honest setting the claim that the server obtains no information about the contacts of undiagnosed users may be incorrect (for example, if the authorities can send to a suspect two tokens $x$ and $x'$ such that $x'=x^2$ then they will be able to identify the suspect as Alice when she performs the protocol).

\subsection*{Open questions}

The trick for this protocol was to ensure correct behaviour not by cumbersome zero-knowledge proofs but by rerandomising in such a way that a malformed token just results in the malefactor seeing random noise.  The most important question for future work is whether a similar trick can be applied to obtain a lightweight DH-based protocol for PSI-CA which is robust against a fully malicious adversary.  This would be extremely desirable because it could easily be added to DP3T-style systems with no changes to the system structure or to the technically-constrained Bluetooth Low Energy tokens.

A second question is whether it is possible for the server, rather than sending all the rerandomised tokens to each user, to instead combine them in some way such that the user can tell whether they included at least one of the special form $(x,x^{s_u})$.  This would be desirable for both performance and privacy reasons, since it would prevent users from learning how many of the tokens sent in by infected individuals were theirs.  If the question was whether they were \emph{all} of the special form then this would be trivial: just multiply together all the tokens componentwise.  Unfortunately we have been unable to find a similar solution for the `disjunctive' task.

\printbibliography
\end{document}